\documentclass[10pt]{article}%

\usepackage{graphicx,amsmath,graphicx,amsfonts,amssymb, amsmath}
\providecommand{\U}[1]{\protect\rule{.1in}{.1in}}
\newtheorem{theorem}{Theorem}
\newtheorem{corollary}{Corollary}
\newtheorem{definition}{Definition}

\newtheorem{example}{Example}
\newtheorem{proposition}{Proposition}
\newtheorem{remark}{Remark}
\newenvironment{proof}[1][Proof]{\noindent\textbf{#1.} }{\hfill $\Box$}

\begin{document}

\title{Properties of palindromes in finite words}
\author{Mira-Cristiana ANISIU\thanks{Tiberiu Popoviciu Institute of Numerical
Analysis, Romanian Academy, Cluj-Napoca, E-mail: \texttt{mira@math.ubbcluj.ro}%
}
\and Valeriu ANISIU\thanks{Babe\c s-Bolyai University of Cluj-Napoca, Faculty of
Mathematics and Computer Science, Department of Mathematics, E-mail:
\texttt{anisiu@math.ubbcluj.ro}}
\and Zolt\'{a}n K\'{A}SA\thanks{Babe\c{s}-Bolyai University of Cluj-Napoca, Faculty
of Mathematics and Computer Science, Department of Computer Science, E-mail:
\texttt{kasa@cs.ubbcluj.ro}}}
\date{}
\maketitle

\begin{abstract}
We present a method which displays all palindromes of a given length from
De Bruijn words of a certain order, and also a recursive one which constructs 
all palindromes of length $n+1$\ from the set of palindromes of length
$n$. We show that the palindrome complexity function, which counts the number
of palindromes of each length contained in a given word, has a different shape
compared with the usual (subword) complexity function. We give upper bounds
for the average number of palindromes contained in all words of length
$n$, and obtain exact formulae for the number of palindromes of length 1 and 2
contained in all words of length $n$.

\end{abstract}

\bigskip{\small \noindent\textbf{AMS 2000 subject classifications:}
68R15\newline} {\small \textbf{Key words and phrases:} finite words,
palindromes, palindrome complexity}

\section{Introduction}

The palindrome complexity of infinite words has been studied by several
authors (see \cite{abcd}, \cite{br}, \cite{aa} and the references therein). Similar
problems related to the number of palindromes are important for finite words
too. One of the reasons is that palindromes occur in DNA sequences (over
$4$\ letters) as well as in protein description (over $20$\ letters), and
their role is under research (\cite{gp}).

Let an alphabet $A$ with $\textrm{card}(A)=q\geq1$ be given. The set of
the words over $A$ will be denoted by $A^*$, and the set of words of length $n$\  by $A^{n}$.  

Given a word $w=w_{1}w_{2}...w_{n},$ the \emph{reversed} of $w$\ is
$\widetilde{w}=w_{n}...w_{2}w_{1}$.\ Denoting by $\varepsilon$\ the empty
word, we put by convention $\widetilde{\varepsilon}=\varepsilon$. The word $w
$\ is a \emph{palindrome} if $\widetilde{w}=w$. We denote by $a^{k}$\ the word
$\underbrace{a\ldots a}_{k \textrm{ times}}$. The set of the subwords of a word $w$\ which are nonempty
palindromes will be denoted by $\mathrm{PAL}(w)$. The (infinite) set of all palindromes
over the alphabet $A$\ is denoted by $\mathrm{PAL}(A)$, while $\mathrm{PAL}_{n}(A)=\mathrm{PAL}(A)\cap A^{n}$.

\section{Storing and generating palindromes}

An old problem asks if, given an alphabet $A$ with $\textrm{card}$($A)=q,$   there exists a shortest word of length $q^{k}+k-1$\ containing all the $q^{k}$ words of length $k$. The answer is affirmative and was given in
\cite{f}, \cite{good}, \cite{br1}. For each $k\in\mathbb{N}$, these words are
called De Bruijn words of order $k$. This property can be proved by means of
the Eulerian cycles in the De Bruijn graph $B_{k-1}$. If a window of length
$k$\ is moved along a De Bruijn word, at each step a different word is seen,
all the $q^{k}$\ words being displayed.

We ask if it is possible to arrange all  palindromes of length $k$\ in a
similar way. The answer is in general no, excepting the case of the two
palindromes $aba...a$\ and $bab...b$ of odd length.

\begin{proposition}
Given a word $w\in A^{n}$\ and $k\geq2$, the following statements are equivalent:

(1) all the subwords of length $k$\ are palindromes;

(2) $n$\ is even, $k=n-1$\ and there exists $a,b\in A$, $a\neq b$\ so that
$w=(ab)^{n/2}$. 

Furthemore, in this case the only palindromes of $w$\ are $(ab)^{n/2-2}a$\ and
$(ba)^{n/2-2}b$.
\end{proposition}

\begin{proof}
Let us consider the first two palindromes $a_{1}a_{2}...a_{k}$\ and
$b_{1}b_{2}...b_{k}$\ such that $a_{2}a_{2}...a_{k}=b_{1}b_{2}...b_{k-1}$,
hence
\[
a_{k-i+1}=a_{i}=b_{i-1}=b_{k-i+2},\ \ i=2,...,k.
\]
It follows
\[%
\begin{array}
[c]{lll}%
i=2 &  & a_{k-1}=a_{2}=b_{1}=b_{k}\\
i=3 &  & a_{k-2}=a_{3}=b_{2}=b_{k-1}\\
i=4 &  & a_{k-3}=a_{4}=b_{3}=b_{k-2}\\
&  & .\ .\ .\ \\
i=k-1 &  & a_{2}=a_{k-1}=b_{k-2}=b_{3}\\
i=k &  & a_{1}=a_{k}=b_{k-1}=b_{2}.
\end{array}
\]
\noindent If $k=2l$, ($l\geq1$)\ we have $b_{2}=a_{1}=a_{3}=...=a_{k-1}$ and
$b_{3}=a_{2}=...=a_{k}$ and $a_{1}a_{2}...a_{k}$ is a palindrome if and only
if $a_{1}=a_{2}=...=a_{k}$, hence $a_{1}a_{2}...a_{k}=a^{k}$; it follows that
$b_{1}b_{2}...b_{k}=a^{k}$\ too, and the two palindromes are equal.

\allowbreak\noindent If $k=2l+1$, we have $b_{2}=a_{1}=a_{3}=...=a_{k}$ and
$b_{3}=a_{2}=...=a_{k-1,}$ hence $a_{1}a_{2}...a_{k}=abab...a$ ($a\neq b$) and
$b_{1}b_{2}...b_{k}=bab...b$. If another palindrome will follow, it must be
again $(ab)^{n/2}$ (equal with the first one).
\end{proof}

\begin{remark}
For $k=1$, the maximum length of a word containing all distinct
palindromes of length $1$ (i.e. letters) exactly once is $n=q$.
\end{remark}

It is obvious that for $k\geq2$ it is not possible to arrange all 
palindromes of length $k$\ in the most compact way. But each palindrome is
determined by the parity of its length and its first $\left\lceil
k/2\right\rceil $\ letters, where $\left\lceil \cdot\right\rceil $\ denotes
the ceil function (which return the smallest integer that is greater than or
equal to a specified number).

\begin{proposition}
All palindromes of length $k$\ can be obtained from a De Bruijn word of
length $q^{\left\lceil k/2\right\rceil }+\left\lceil k/2\right\rceil -1$.
\end{proposition}

\begin{proof}
The De Bruijn word contains all different words of length $\left\lceil
k/2\right\rceil $. Each such word $a_{1}...a_{\left\lceil k/2\right\rceil }%
$\ can be extended to a palindrome by symmetry, for $k$\ even, and by taking
$a_{\left\lceil k/2\right\rceil +1}=a_{\left\lceil k/2\right\rceil -1}$,
...,$a_{k}=a_{1}$, for $k$\ odd.
\end{proof}

\begin{example}
Let $k=3$, $q=3$ and the De Bruijn word of order $\left\lceil k/2\right\rceil $ $=2$ $w_{1}= 0221201100$. From each word of length
$2$\ which appears in the given De Bruijn word, we obtain the corresponding
palindrome of length $k=3$:%
\[%
\begin{array}
[c]{ccc}%
02 & \rightarrow & 020\\
22 & \rightarrow & 222\\
21 & \rightarrow & 212\\
12 & \rightarrow & 121\\
20 & \rightarrow & 202\\
01 & \rightarrow & 010\\
11 & \rightarrow & 111\\
10 & \rightarrow & 101\\
00 & \rightarrow & 000.
\end{array}
\]
Let $k=4$, $q=2$ and the De Bruijn word of order $\left\lceil k/2\right\rceil =2$ $w_{2}=01100$. From each word of length $2$\ contained in
$01100$\ we obtain by symmetry the corresponding palindrome of length $k=4$:%
\[%
\begin{array}
[c]{ccc}%
01 & \rightarrow & 0110\\
11 & \rightarrow & 1111\\
10 & \rightarrow & 1001\\
00 & \rightarrow & 0000.
\end{array}
\]
\end{example}

There are several algorithms which construct De Bruijn words, for example, in
\cite{m}, \cite{r}, \cite{fr} and  \cite{g}.

We can generate recursively all  palindromes of length $n$, $n\in\mathbb{N}$, using the
difference representation. This is based on the following proposition.

\begin{proposition}
If  $w_1, w_2,\; \ldots ,\; w_p $
are all binary $(A=\{0,1\})$ palindromes of length $n$, $\textrm{ where } p=2^{\left\lceil \frac{n}{2}\right\rceil }$,  $n\ge 1$, then 
\[2w_1, 2w_2,\; \ldots \;, 2w_p, 2^{n+1}+1+w_1, 2^{n+1}+1+w_2,\;\ldots ,\;2^{n+1}+ 1+w_p\]
are all palindromes of length $n+2$.
\end{proposition}
\begin{proof}
If $w$ is a binary palindrome of length $n$, then $0w0$ and $1w1$ will be palindromes too, and the only palindromes of length $n+2$ which contains $w$ as a subword, which proves the proposition.
\end{proof}

In order to generate all binary palindromes of a given length let us begin with an example
considering all binary palindromes of length 3 and 4 and their decimal representation:

$%
\begin{array}
[c]{lllr}%
000 \; & 0 \quad\qquad & 0000 \; & 0\\
010 & 2 & 0110 & 6\\
101 & 5 & 1001 & 9\\
111 & 7 & 1111 & 15
\end{array}
$

The sequence of palindromes in increasing order based on their decimal value for
a given length can be represented by their differences. The difference
representation of the sequence 0, 2, 5, 7 is 2, 3, 2 ($2-0=2$, $5-2=3$,
$7-5=2$), and the difference representation of the sequence 0, 6, 9, 15 is 6,
3, 6. A difference representation is always a simmetric sequence and the
corresponding sequence of palindromes in decimal can be obtained by successive
addition beginning with 0: $\mathbf{0}+6=\mathbf{6}$, $\mathbf{6}%
+3=\mathbf{9}$, $\mathbf{9}+6=\mathbf{15}$. By direct computation we obtain
the following difference representation of palindromes for length $n\leq8$.

\medskip\noindent%
\begin{tabular}
[c]{@{}c@{}rrrrrrrrrrrrrrr}%
$n$ &  &  &  &  &  &  &  &  &  &  &  &  &  &  & \\
\emph{1} & 1 &  &  &  &  &  &  &  &  &  &  &  &  &  & \\
\emph{2} & 3 &  &  &  &  &  &  &  &  &  &  &  &  &  & \\
\emph{3} & 2 & 3 & 2 &  &  &  &  &  &  &  &  &  &  &  & \\
\emph{4} & 6 & 3 & 6 &  &  &  &  &  &  &  &  &  &  &  & \\
\emph{5} & 4 & 6 & 4 & 3 & 4 & 6 & 4 &  &  &  &  &  &  &  & \\
\emph{6} & 12 & 6 & 12 & 3 & 12 & 6 & 12 &  &  &  &  &  &  &  & \\
\emph{7} & 8 & 12 & 8 & 6 & 8 & 12 & 8 & 3 & 8 & 12 & 8 & 6 & 8 & 12 & 8\\
\emph{8} & \;\,24 & 12 & 24 & 6 & 24 & 12 & 24 & 3 & 24 & 12 & 24 & 6 & 24 &
\;12 & \;24
\end{tabular}

\medskip We easily can generalize and prove by induction that
 the difference representations can be obtained as follows.

For $n=2k$ we have the difference representation:
\[
a_{1}, a_{2},\;  \ldots, \; a_{2^{k}-1},
\]
from which the difference representation for $2k+1$ is:
\[
2^{k}, a_{1}, 2^{k}, a_{2},2^{k},\; \ldots, 2^{k},\; a_{2^{k}-1},2^{k}.
\]

For $n=2k+1$ we have the difference representation:
\[
2^{k}, a_{1}, 2^{k}, a_{2},2^{k},\; \ldots, 2^{k},\; a_{2^{k}-1},2^{k},
\]
from which the difference representation for $2k+2$ is:
\[
3\cdot2^{k}, a_{1}, 3\cdot2^{k}, a_{2}, 3\cdot2^{k},\; \ldots, \; 3\cdot2^{k},
a_{2^{k}-1},3\cdot2^{k}.
\]

This representation can be generalized for $q\ge2$. The number of palindromes
in this case is $q^{\left\lceil \frac{n}{2} \right\rceil }$.

For $n=2k$ we have the difference representation:
\[
a_{1}, a_{2},\; \ldots, \; a_{q^{k}-1},
\]
from which the difference representation for $2k+1$ is:
\[
\underbrace{q^{k},\ldots,  q^{k}}_{q-1 \text{ times}},\; a_{1}, \;
\underbrace{q^{k},\ldots, q^{k}}_{q-1 \text{ times}}, \; a_{2}, \;\underbrace
{q^{k},\ldots,  q^{k}}_{q-1 \text{ times}}, \; \ldots, \; \underbrace{q^{k},\ldots, 
q^{k}}_{q-1 \text{ times}}, \; a_{q^{k}-1},\; \underbrace{q^{k},\ldots,  q^{k}%
}_{q-1 \text{ times}}.
\]

For $n=2k+1$ we have the difference representation:
\[
\underbrace{q^{k},\ldots, q^{k}}_{q-1\text{ times}},\;a_{1},\;\underbrace
{q^{k},\ldots, q^{k}}_{q-1\text{ times}},\;a_{2},\;\ldots,\;a_{q^{k}%
-1},\;\underbrace{q^{k},\ldots, q^{k}}_{q-1\text{ times}},
\]
from which the difference representation for $2k+2$ is:
\[
\underbrace{(q+1)q^{k},\ldots, (q+1)q^{k}}_{q-1\text{ times}},\;a_{1}%
,\;\underbrace{(q+1)q^{k},\ldots, (q+1)q^{k}}_{q-1\text{ times}},\;a_{2},
\]%
\[
\ldots,\;\underbrace{(q+1)q^{k},\;\ldots,\;(q+1)q^{k}}_{q-1\text{ times}%
},\;a_{q^{k}-1},\;\underbrace{(q+1)q^{k},\ldots, (q+1)q^{k}}_{q-1\text{ times}}.
\]

\section{The shape of the palindrome complexity functions}

For an infinite sequence $U,$ the \textit{(subword)} \textit{complexity
function} $p_{U}:\mathbb{N}\longrightarrow\mathbb{N}$ (defined in \cite{mh} as
the \textit{block growth}, then named \textit{subword complexity} in
\cite{elr}) is given by $p_{U}(n)=$ $\textrm{card}(F(U)\cap A^{n})$ for
$n\in\mathbb{N},$ where $F(U)$ is the set of all finite subwords (factors) of $U$. 
Therefore the complexity function  maps each nonnegative number $n$ to the number of
subwords of length $n$ of $U;$ it verifies the iterative equation
\begin{equation}
p_{U}(n+1)=p_{U}(n)+\sum\limits_{j=2}^{q}(j-1)s(j,n),\label{f1}%
\end{equation}
$s(j,n)$ being the cardinal of the set of the subwords in $U$ having the length
$n$ and the right valence $j.$ A subword $u\in U$\ \emph{has the right
valence} $j$ if there are $j$\ and only $j$\ distinct letters $x_{i}$\ such
that $ux_{i}\in F(U)$, $1\leq i\leq j$.

For a finite word $w$ of length $n,$ the \textit{complexity function}
$p_{w}:\mathbb{N}\longrightarrow\mathbb{N}$ given by $p_{w}(k)=$
$\textrm{card}(F(w)\cap A^{k}),$ $k\in\mathbb{N},$ has the property
that $p_{w}(k)=0$ for $k>n.$ The corresponding iterative equation is
\begin{equation}
p_{w}(k+1)=p_{w}(k)+\sum\limits_{j=2}^{q}(j-1)s(j,k)-s_{0}(k),\label{f2}%
\end{equation}
where $s_{0}(k)=s(0,k)\in\{0,1\}$ stands for the cardinal of the set of
subwords $v$\ (suffixes of $w$ of length $k$) which cannot be continued as
$vx\in F(w)$,\ $x\in A$. We can write (\ref{f2}) in a condensed form
\begin{equation}
p_{w}(k+1)=p_{w}(k)+\sum\limits_{j=0}^{q}(j-1)s(j,k).\label{f3}%
\end{equation}

The above relations have their correspondents in terms of left extensions of
the subwords.

For an infinite sequence $U,$ the complexity function $p_{U}$ is
nondecreasing; more than that, if there exists $m\in\mathbb{N}$ such that
$p_{U}(m+1)=p_{U}(m),$ then $p_{U}$ is constant for $n\geq m$.

The complexity function for a finite word $w$ of length $n$ has a different
behaviour, because of $p_{w}(n)=1$ (there is a unique subword of length $n,$
namely $w$). It was proved (\cite{h}, \cite{adl}, \cite{l-s}, \cite{ac}) that the shape of the complexity function is trapezoidal.

\begin{theorem}
Given a finite word $w$ of length $n,$ there are three intervals of
monotonicity for $p_{w}$: $[0,J],$ $[J,M]$ and $[M,n]$; the function increases
at first, is constant and then decreases with the slope $-1$.
\end{theorem}

The \emph{palindrome complexity function} of a finite or infinite word $w$\ is
given by $\mathrm{pal}_{w}:\mathbb{N}\longrightarrow\mathbb{N}$, $\mathrm{pal}_{w}(k)=$
$\textrm{card}(\mathrm{PAL}(w)\cap A^{k}),$ $k\in\mathbb{N}$. Obviously,
\begin{equation}
\mathrm{pal}_{w}(k)\leq p_{w}(k),\ k\in\mathbb{N},\label{ac1}%
\end{equation}
and for finite words of length $\left\vert w\right\vert =n$,%
\begin{equation}
\mathrm{pal}_{w}(k)\leq\min\left\{  q^{\left\lceil k/2\right\rceil },n-k+1\right\}
,\ k\in\{0,...,n\}.\label{ac2}%
\end{equation}
The palindrome $u\in \mathrm{PAL}(w)$\ \emph{has the palindrome valence} $j$ if there
are $j$\ and only $j$\ distinct letters $x_{i}$\ such that $x_{i}ux_{i}\in
\mathrm{PAL}(w)$, $1\leq i\leq j$. We denote by
\begin{equation}
s_{p}(j,k)=\textrm{card}\left\{  u\in(\mathrm{PAL}(w)\cap A^{k}):u\text{ has the
palindrome valence }j\right\}  ,\label{ac3}%
\end{equation}
and by $s_{p}(0,k)$ the cardinal of the set of subwords $v\in \mathrm{PAL}(w)\cap A^{k}
$\ (not necessarily suffixes or prefixes of $w$) which cannot be continued as
$xvx\in \mathrm{PAL}(w)$,\ $x\in A$.

The palindrome complexity function of finite or infinite words satisfies the
iterative equation
\begin{equation}
\mathrm{pal}_{w}(k+2)=\mathrm{pal}_{w}(k)+\sum\limits_{j=0}^{q}(j-1)s_{p}(j,k).\label{ac4}%
\end{equation}
Due to the fact that the number of even palindromes is not directly related to
that of odd ones, we do not expect that $\mathrm{pal}_{w}$\ is of trapezoidal shape, as
it was the case for the subword complexity function $p_{w}$.

For this reason we define the \emph{odd}, respectively\emph{\ even palindrome
complexity function} as the restrictions of $\mathrm{pal}_{w}\ $to odd, respectively
even integers: $\mathrm{pal}_{w}^{o}:2\mathbb{N}+1\rightarrow\mathbb{N},$ $\mathrm{pal}_{w}%
^{o}(k)=\mathrm{pal}_{w}(k);\ \mathrm{pal}_{w}^{e}:2\mathbb{N}\rightarrow\mathbb{N},$
$\mathrm{pal}_{w}^{e}(k)=\mathrm{pal}_{w}(k)$.

These functions have a trapezoidal form for short words; nevertheless, this is
not true in general, as the following examples show.

\begin{example}
\textrm{\label{ex1}The word $w_1=1010^{5}1^{2}0^{7}10$ with $\left\vert
w_{1}\right\vert =19$ has $\mathrm{pal}_{w_{1}}^{o}(1)=2$, $\mathrm{pal}_{w_{1}}^{o}(3)=3$,
$\mathrm{pal}_{w_{1}}^{o}(5)=1$, $\mathrm{pal}_{w_{1}}^{o}(7)=2$, $\mathrm{pal}_{w_{1}}^{o}(9)=1$. } (see Fig. 1.)
\end{example}

\begin{example}
\textrm{\label{ex2}The word $w_{2}=1^{4}0^{6}10^{8}1^{2}0$ with $\left\vert
w_{2}\right\vert =22$ has $\mathrm{pal}_{w_{2}}^{e}(2)=2$, $\mathrm{pal}_{w_{2}}^{e}(4)=3$,
$\mathrm{pal}_{w_{2}}^{e}(6)=1$, $\mathrm{pal}_{w_{2}}^{e}(8)=2$, $\mathrm{pal}_{w_{2}}^{e}(10)=1$. } (see Fig. 1.)
\end{example}

\begin{figure}
\includegraphics[width=\textwidth]{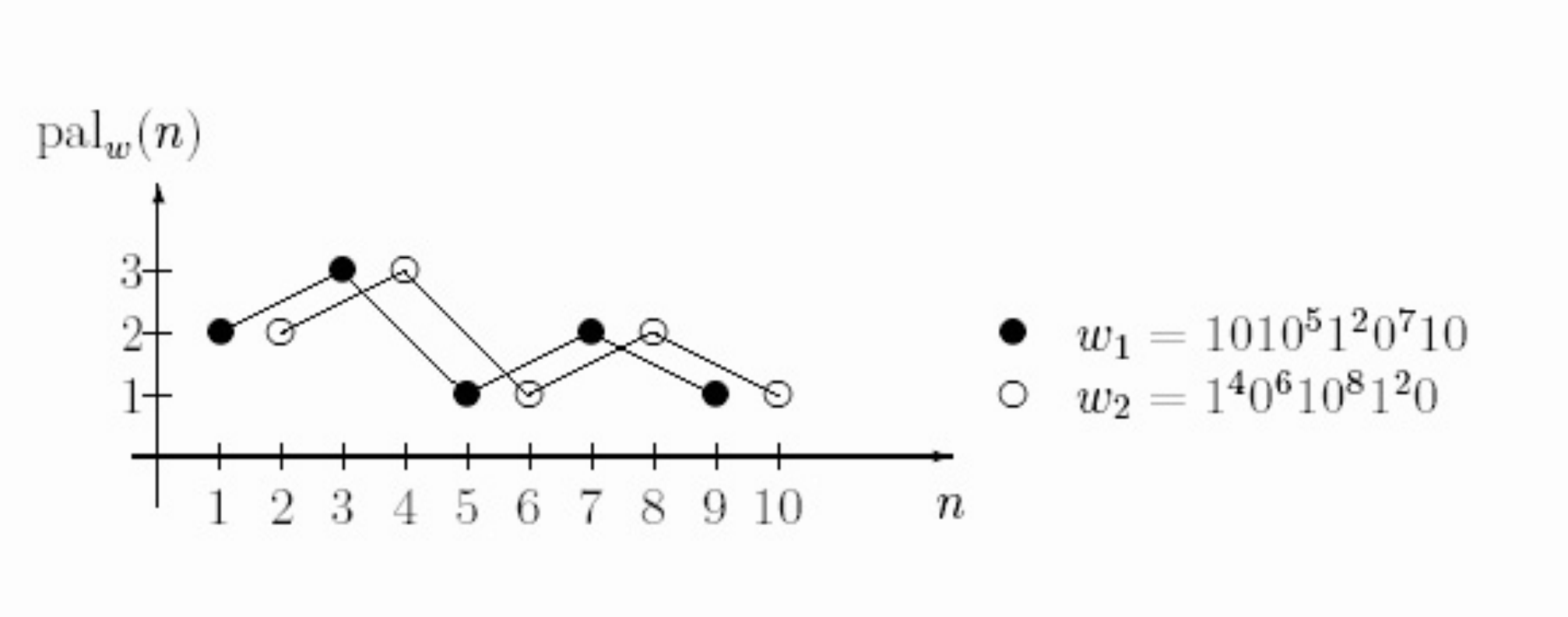}
\caption{Odd and even palindrome complexity function}
\end{figure}

\begin{remark}
The palindrome complexity for infinite words is not nondecreasing, as the
usual complexity function is. Indeed, we can continue the word in Example
\ref{ex1} with $11001100$..., and its odd palindrome complexity function will
be as that for $w_{1},$\ and then equal to $0$ for $k\geq11.$ Similarly, we
can continue $w_{2}$\ in Example \ref{ex2} with $1010$... to obtain an
infinite word with the even palindrome complexity of $w_{2}$\ till $k=10$ and
equal to $0$ for $k\geq12.$
\end{remark}

\section{Average number of palindromes}

We consider an alphabet $A$ with $q\geq2$\ letters.

\begin{definition}
We define the \textbf{total palindrome complexity} $P$\ by
\begin{equation}
P(w)=%
%TCIMACRO{\dsum _{n=1}^{\left\vert w\right\vert }}%
%BeginExpansion
{\displaystyle\sum_{n=1}^{\left\vert w\right\vert }}
%EndExpansion
\mathrm{pal}_{w}(n),\label{p1}%
\end{equation}
where $w$\ is a word of length $\left\vert w\right\vert $,\ and $\mathrm{pal}_{w}%
(n)$\ denotes the number of distinct palindromes of length $n$\ which are
nonempty subwords of $w$. 
\end{definition}

Because he set of the nonempty palindromes in $w$\ is denoted
by $\mathrm{PAL}(w)$, we can write also $P(w)=\textrm{card}(\mathrm{PAL}(w))$.

\begin{definition}
The \textbf{average number of palindromes} $M_{q}(n)$\ contained in all 
words of length $n$\ is defined by
\begin{equation}
M_{q}(n)=\frac{%
%TCIMACRO{\dsum \limits_{w\in A^{n}}}%
%BeginExpansion
{\displaystyle\sum\limits_{w\in A^{n}}}
%EndExpansion
P(w)}{q^{n}}.\label{p2}%
\end{equation}

\end{definition}

We can give the following upper estimate for $M_{q}(n)$.

\begin{theorem}
\label{tp}For $n\in\mathbb{N}$,\ the average number of palindromes contained
in the words of length $n$\ satisfies the inequalities
\begin{equation}%
\begin{array}
[c]{l}%
M_{q}(n)\leq\dfrac{q^{-\left(  n-1\right)  /2}(q+3)+2n(q-1)+q^{3}-2q^{2}%
-2q-1}{(q-1)^{2}},\text{ for }n\text{\ odd,}\\
M_{q}(n)\leq\dfrac{q^{-n/2}(3q+1)+2n(q-1)+q^{3}-2q^{2}-2q-1}{(q-1)^{2}}%
,\quad\;\text{for }n\text{\ even.}%
\end{array}
\label{p4}%
\end{equation}

\end{theorem}

\begin{proof}
We have
\[%
\begin{array}
[c]{ll}%
%TCIMACRO{\dsum \limits_{w\in A^{n}}}%
%BeginExpansion
{\displaystyle\sum\limits_{w\in A^{n}}}
%EndExpansion
P(w) & =%
%TCIMACRO{\dsum \limits_{w\in A^{n}}}%
%BeginExpansion
{\displaystyle\sum\limits_{w\in A^{n}}}
%EndExpansion
\
%TCIMACRO{\dsum \limits_{\pi\in \mathrm{PAL}(w)}}%
%BeginExpansion
{\displaystyle\sum\limits_{\pi\in \mathrm{PAL}(w)}}
%EndExpansion
1=%
%TCIMACRO{\dsum \limits_{w\in A^{n}}}%
%BeginExpansion
{\displaystyle\sum\limits_{w\in A^{n}}}
%EndExpansion
\
%TCIMACRO{\dsum \limits_{k=1}^{n}}%
%BeginExpansion
{\displaystyle\sum\limits_{k=1}^{n}}
%EndExpansion
\
%TCIMACRO{\dsum \limits_{\pi\in \mathrm{PAL}(w)\cap A^{k}}}%
%BeginExpansion
{\displaystyle\sum\limits_{\pi\in \mathrm{PAL}(w)\cap A^{k}}}
%EndExpansion
1\\
& \leq%
%TCIMACRO{\dsum \limits_{w\in A^{n}}}%
%BeginExpansion
{\displaystyle\sum\limits_{w\in A^{n}}}
%EndExpansion
\
%TCIMACRO{\dsum \limits_{\pi\in \mathrm{PAL}(w)\cap A^{1}}}%
%BeginExpansion
{\displaystyle\sum\limits_{\pi\in \mathrm{PAL}(w)\cap A^{1}}}
%EndExpansion
1+%
%TCIMACRO{\dsum \limits_{k=2}^{n}}%
%BeginExpansion
{\displaystyle\sum\limits_{k=2}^{n}}
%EndExpansion
\
%TCIMACRO{\dsum \limits_{\pi\in \mathrm{PAL}_{k}(A)}}%
%BeginExpansion
{\displaystyle\sum\limits_{\pi\in \mathrm{PAL}_{k}(A)}}
%EndExpansion
\
%TCIMACRO{\dsum \limits_{\substack{w\in A^{n} \\\pi\in \mathrm{PAL}(w)\cap A^{k}}}}%
%BeginExpansion
{\displaystyle\sum\limits_{\substack{w\in A^{n} \\\pi\in \mathrm{PAL}(w)\cap A^{k}}}}
%EndExpansion
1,
\end{array}
\]
and
\begin{equation}%
%TCIMACRO{\dsum \limits_{w\in A^{n}}}%
%BeginExpansion
{\displaystyle\sum\limits_{w\in A^{n}}}
%EndExpansion
\
%TCIMACRO{\dsum \limits_{\pi\in \mathrm{PAL}(w)\cap A^{1}}}%
%BeginExpansion
{\displaystyle\sum\limits_{\pi\in \mathrm{PAL}(w)\cap A^{1}}}
%EndExpansion
1\leq qq^{n}=q^{n+1}.\label{pp1}%
\end{equation}
For a fixed palindrome $\pi$, with $\left\vert \pi\right\vert =k$, the number
of the words of length $n$\ in which it appears as a subword at position
$i$\ ($1\leq i\leq n-k+1$) is $q^{n-k}$. But the position $i$\ is arbitrary,
so that there are at most $(n-k+1)q^{n-k}$\ words in which $\pi$\ is a
subword, these words being not necessarily distinct. It follows that
\[%
%TCIMACRO{\dsum \limits_{w\in A^{n}}}%
%BeginExpansion
{\displaystyle\sum\limits_{w\in A^{n}}}
%EndExpansion
P(w)\leq q^{n+1}+%
%TCIMACRO{\dsum \limits_{k=2}^{n}}%
%BeginExpansion
{\displaystyle\sum\limits_{k=2}^{n}}
%EndExpansion
\
%TCIMACRO{\dsum \limits_{\pi\in \mathrm{PAL}_{k}(A)}}%
%BeginExpansion
{\displaystyle\sum\limits_{\pi\in \mathrm{PAL}_{k}(A)}}
%EndExpansion
\ (n-k+1)q^{n-k}.
\]
The number of the palindromes of length $k$\ is $q^{\left\lceil
k/2\right\rceil }$, therefore
\[%
%TCIMACRO{\dsum \limits_{w\in A^{n}}}%
%BeginExpansion
{\displaystyle\sum\limits_{w\in A^{n}}}
%EndExpansion
P(w)\leq q^{n+1}+%
%TCIMACRO{\dsum \limits_{k=2}^{n}}%
%BeginExpansion
{\displaystyle\sum\limits_{k=2}^{n}}
%EndExpansion
(n-k+1)q^{n-k+\left\lceil k/2\right\rceil }%
\]
and
\[
M_{q}(n)\leq q+%
%TCIMACRO{\dsum \limits_{k=2}^{n}}%
%BeginExpansion
{\displaystyle\sum\limits_{k=2}^{n}}
%EndExpansion
(n-k+1)q^{-k+\left\lceil k/2\right\rceil }.
\]
We split the sum according to $k=2j,$\ $j=1,...,\left\lfloor n/2\right\rfloor
,$\ respectively $k=2j+1,$\ $j=1,...,\left\lfloor (n-1)/2\right\rfloor $, and
obtain
\[
M_{q}(n)\leq q+%
%TCIMACRO{\dsum \limits_{j=1}^{\left\lfloor n/2\right\rfloor }}%
%BeginExpansion
{\displaystyle\sum\limits_{j=1}^{\left\lfloor n/2\right\rfloor }}
%EndExpansion
(n-2j+1)q^{-j}+%
%TCIMACRO{\dsum \limits_{j=1}^{\left\lfloor (n-1)/2\right\rfloor }}%
%BeginExpansion
{\displaystyle\sum\limits_{j=1}^{\left\lfloor (n-1)/2\right\rfloor }}
%EndExpansion
(n-2j)q^{-j}.
\]
Making use of $%
%TCIMACRO{\dsum \limits_{j=1}^{s}}%
%BeginExpansion
{\displaystyle\sum\limits_{j=1}^{s}}
%EndExpansion
q^{-j}=(1-q^{-s})/(q-1)$\ and $%
%TCIMACRO{\dsum \limits_{j=1}^{s}}%
%BeginExpansion
{\displaystyle\sum\limits_{j=1}^{s}}
%EndExpansion
jq^{-j}=(q-q^{1-s}(s+1)+sq^{-s})/(q-1)^{2}$, it follows that $M_{q}%
(n)$\ satisfies the inequalities in (\ref{p4}).
\end{proof}

\begin{corollary}
\label{c1}The following inequality holds
\begin{equation}
\limsup_{n\rightarrow\infty}\frac{M_{q}(n)}{n}\leq\frac{2}{q-1}.\label{p5}%
\end{equation}

\end{corollary}

\begin{proof}%
\[%
\begin{array}
[c]{l}%
\limsup\limits_{n\rightarrow\infty}\dfrac{M_{q}(n)}{n}=\max\left\{
\limsup\limits_{n\rightarrow\infty}\dfrac{M_{q}(2n+1)}{2n+1},\ \limsup
\limits_{n\rightarrow\infty}\dfrac{M_{q}(2n)}{2n}\right\} \\
\leq\max\left\{  \lim\limits_{n\rightarrow\infty}\left(  \dfrac{q^{-n}%
(q+3)+2\left(  2n+1\right)  (q-1)+q^{3}-2q^{2}-2q-1}{(q-1)^{2}}\right)
\dfrac{1}{2n+1},\right. \\
\left.  \lim\limits_{n\rightarrow\infty}\left(  \dfrac{q^{-n}%
(3q+1)+4n(q-1)+q^{3}-2q^{2}-2q-1}{(q-1)^{2}}\right)  \dfrac{1}{2n}\right\}
=\dfrac{2}{q-1}.
\end{array}
\]

\end{proof}

We are interested in finding how large is the average number of palindromes
contained in the words of length $n$\ compared to the length $n$. The
numerical estimations done for small values of $n$ show that $M_{q}(n)$\ is
comparable to $n$,\ but Corollary \ref{c1} allows us to show that for $q\geq
4$\ this does not hold.

\begin{corollary}
\label{c2}For an alphabet with $q\geq4$ letters,
\begin{equation}
\limsup_{n\rightarrow\infty}\frac{M_{q}(n)}{n}<1\text{.}\label{p7}%
\end{equation}

\end{corollary}

In the proof of Theorem \ref{tp} we have used the rough inequality
(\ref{pp1}), which was sufficient to prove the result. In fact, it is not
difficult to calculate exactly
\begin{equation}
S_{n,p}=%
%TCIMACRO{\dsum \limits_{w\in A^{n}}}%
%BeginExpansion
{\displaystyle\sum\limits_{w\in A^{n}}}
%EndExpansion
\
%TCIMACRO{\dsum \limits_{\pi\in \mathrm{PAL}(w)\cap A^{p}}}%
%BeginExpansion
{\displaystyle\sum\limits_{\pi\in \mathrm{PAL}(w)\cap A^{p}}}
%EndExpansion
1\text{ for }p=1,2.\label{pp}%
\end{equation}
This result has intrinsic importance.

\begin{theorem}
\label{prop1}The number of occurrences of the palindromes of length $1$,
respectively $2$, in all  words of length $n$\ (counted once if a
palindrome appears in a word, and once again if it appears in another one) is
given by
\begin{equation}
S_{n,1}=q^{n+1}-q\left(  q-1\right)  ^{n},\label{p8}%
\end{equation}
respectively by
\begin{equation}%
\begin{array}
[c]{l}%
S_{n,2}=q^{n+1}-\dfrac{q}{(q-1)\sqrt{q^{2}+q-3}}\left(  \left(  \dfrac
{q-1+\sqrt{q^{2}+q-3}}{2}\right)  ^{n+2}\right. \\
\left.  -\left(  \dfrac{q-1-\sqrt{q^{2}+q-3}}{2}\right)  ^{n+2}\right)  .
\end{array}
\label{p9}%
\end{equation}

\end{theorem}

\begin{proof}
We use Iverson's convention \cite{gkp}%
\[
\left[  \alpha\right]  =\left\{
\begin{array}
[c]{l}%
1,\text{ if }\alpha\text{\ is true}\\
0,\text{ if }\alpha\text{\ is false}%
\end{array}
\right.
\]
and obtain
\[
S_{n,1}=%
%TCIMACRO{\dsum \limits_{w\in A^{n}}}%
%BeginExpansion
{\displaystyle\sum\limits_{w\in A^{n}}}
%EndExpansion
\
%TCIMACRO{\dsum \limits_{a\in A}}%
%BeginExpansion
{\displaystyle\sum\limits_{a\in A}}
%EndExpansion
\left[  a\text{ in }w\right]  =q%
%TCIMACRO{\dsum \limits_{w\in A^{n}}}%
%BeginExpansion
{\displaystyle\sum\limits_{w\in A^{n}}}
%EndExpansion
\left[  a_{1}\text{ in }w\right]  ,
\]
where $a_{1}$\ is a fixed letter of the alphabet $A$. Then
\[
S_{n,1}=q%
%TCIMACRO{\dsum \limits_{w\in A^{n}}}%
%BeginExpansion
{\displaystyle\sum\limits_{w\in A^{n}}}
%EndExpansion
\left[  a_{1}\text{ in }w\right]  =q\left(  q^{n}-%
%TCIMACRO{\dsum \limits_{w\in A^{n}}}%
%BeginExpansion
{\displaystyle\sum\limits_{w\in A^{n}}}
%EndExpansion
\left[  a_{1}\text{ not in }w\right]  \right)  =q^{n+1}-q\left(  q-1\right)
^{n}.
\]

We proceed similarly to calculate $S_{n,2}=%
%TCIMACRO{\dsum \limits_{w\in A^{n}}}%
%BeginExpansion
{\displaystyle\sum\limits_{w\in A^{n}}}
%EndExpansion
\
%TCIMACRO{\dsum \limits_{\pi\in \mathrm{PAL}(w)\cap A^{2}}}%
%BeginExpansion
{\displaystyle\sum\limits_{\pi\in \mathrm{PAL}(w)\cap A^{2}}}
%EndExpansion
1$ and obtain
\[
S_{n,2}=%
%TCIMACRO{\dsum \limits_{w\in A^{n}}}%
%BeginExpansion
{\displaystyle\sum\limits_{w\in A^{n}}}
%EndExpansion
\
%TCIMACRO{\dsum \limits_{a\in A}}%
%BeginExpansion
{\displaystyle\sum\limits_{a\in A}}
%EndExpansion
\left[  aa\text{ in }w\right]  =q%
%TCIMACRO{\dsum \limits_{w\in A^{n}}}%
%BeginExpansion
{\displaystyle\sum\limits_{w\in A^{n}}}
%EndExpansion
\left[  a_{1}a_{1}\text{ in }w\right]  ,
\]
where $a_{1}$\ is again a fixed letter of the alphabet $A$. We denote
$\varphi(n):=%
%TCIMACRO{\dsum \limits_{w\in A^{n}}}%
%BeginExpansion
{\displaystyle\sum\limits_{w\in A^{n}}}
%EndExpansion
\left[  a_{1}a_{1}\text{ in }w\right]  $, for which $\varphi(2)=1$\ and
$\varphi(3)=2q-1$. It is easier to establish a recurrence formula for
$\psi(n)=q^{n}-\varphi(n)=%
%TCIMACRO{\dsum \limits_{w\in A^{n}}}%
%BeginExpansion
{\displaystyle\sum\limits_{w\in A^{n}}}
%EndExpansion
\left[  a_{1}a_{1}\text{ not in }w\right]  $. The number $\psi(n)$\ is
obtained from:

- the number $(q-1)\psi(n-1)$\ of words which do not end in $a_{1}$\ and have
not $a_{1}a_{1}$\ in their first $n-1$\ positions;

- the number $(q-1)\psi(n-2)$\ of words which end in $a_{1}$,\ have the $n-1
$\ position occupied by one of the other $q-1$\ letters and have not
$a_{1}a_{1}$\ in the first $n-2$\ positions.

It follows that $\psi$\ satisfies the recurrence formula
\begin{equation}
\psi(n)=(q-1)(\psi(n-1)+\psi(n-2)),\label{p10}%
\end{equation}
with $\psi(2)=q^{2}-1$ and $\psi(3)=q^{3}-2q+1$. Its solution is
\[%
\begin{array}
[c]{c}%
\psi(n)=\dfrac{1}{(q-1)\sqrt{q^{2}+q-3}}\left(  \left(  \dfrac{q-1+\sqrt
{q^{2}+q-3}}{2}\right)  ^{n+2}\right. \\
\left.  -\left(  \dfrac{q-1-\sqrt{q^{2}+q-3}}{2}\right)  ^{n+2}\right)
\end{array}
\]
and (\ref{p9}) follows from the fact that
\begin{equation}
S_{n,2}=q\left(  q^{n}-\psi(n)\right)  .\label{p10b}%
\end{equation}

\end{proof}

The expression of $S_{n,2}$ from (\ref{p9}) allows us to improve Corollary
\ref{c1}.

\begin{corollary}
\label{c3}The following inequality holds
\begin{equation}
\limsup_{n\rightarrow\infty}\frac{M_{q}(n)}{n}\leq\frac{q+1}{q\left(
q-1\right)  }.\label{p11}%
\end{equation}

\end{corollary}

\begin{proof}
Taking into account the inequality%

\[
{\displaystyle\sum\limits_{w\in A^{n}}} \ {\displaystyle\sum\limits_{\pi\in
\mathrm{PAL}(w)\cap A^{1}}} 1\leq qq^{n}=q^{n+1},
\]
and (\ref{p10b}), we get
\[%
\begin{array}
[c]{ll}%
M_{q}(n) & \leq\dfrac{1}{q^{n}}\left(  S_{n,1}+S_{n,2}+%
%TCIMACRO{\dsum \limits_{k=3}^{n}}%
%BeginExpansion
{\displaystyle\sum\limits_{k=3}^{n}}
%EndExpansion
\
%TCIMACRO{\dsum \limits_{\pi\in \mathrm{PAL}(A^{*})\cap A^{k}}}%
%BeginExpansion
{\displaystyle\sum\limits_{\pi\in \mathrm{PAL}(A)\cap A^{k}}}
%EndExpansion
\ (n-k+1)q^{n-k}\right) \\
& \leq q\left(  2-\dfrac{\psi(n)}{q^{n}}\right)  +%
%TCIMACRO{\dsum \limits_{k=3}^{n}}%
%BeginExpansion
{\displaystyle\sum\limits_{k=3}^{n}}
%EndExpansion
(n-k+1)q^{-k+\left\lfloor (k+1)/2\right\rfloor }.
\end{array}
\]
But $0<\!\left(  q-1\!+\!\sqrt{q^{2}+q-3}\right)  /2<q$ and $-1\!<\!\left(
q-1\!-\!\sqrt{q^{2}+q-3}\right)  /2<0$ for $q\geq2$, hence $\lim
\limits_{n\rightarrow\infty}\psi(n)/q^{n}=0$. Then
\[%
\begin{array}
[c]{ll}%
\limsup\limits_{n\rightarrow\infty}\dfrac{M_{q}(n)}{n} & \leq\lim
\limits_{n\rightarrow\infty}\dfrac{1}{n}%
%TCIMACRO{\dsum \limits_{k=3}^{n}}%
%BeginExpansion
{\displaystyle\sum\limits_{k=3}^{n}}
%EndExpansion
(n-k+1)q^{-k+\left\lfloor (k+1)/2\right\rfloor }\leq%
%TCIMACRO{\dsum \limits_{k=3}^{\infty}}%
%BeginExpansion
{\displaystyle\sum\limits_{k=3}^{\infty}}
%EndExpansion
q^{-k+\left\lfloor (k+1)/2\right\rfloor }\\
& =%
%TCIMACRO{\dsum \limits_{i=1}^{\infty}}%
%BeginExpansion
{\displaystyle\sum\limits_{i=1}^{\infty}}
%EndExpansion
q^{-2i-1+i+1}+%
%TCIMACRO{\dsum \limits_{i=2}^{\infty}}%
%BeginExpansion
{\displaystyle\sum\limits_{i=2}^{\infty}}
%EndExpansion
q^{-2i+i}=-\dfrac{1}{q}+2%
%TCIMACRO{\dsum \limits_{i=1}^{\infty}}%
%BeginExpansion
{\displaystyle\sum\limits_{i=1}^{\infty}}
%EndExpansion
q^{-i}=\dfrac{q+1}{q\left(  q-1\right)  }.
\end{array}
\]

\end{proof}

\begin{corollary}
\label{c4}The inequality (\ref{p7}) holds for $q=3$\ too.
\end{corollary}

It seems that (\ref{p7}) holds also for\ $q=2.$ Using a computer program we
obtained some values for the terms of the sequence $M^{\ast}(n)=M_{2}(n)/n$,
$n\geq2$. The first values are: $M^{\ast}(n)=1,$ $n=2,\ldots,7$; $M^{\ast
}(8)=0.99750$; $M^{\ast}(9)=0.98550$, which were close to $1$. We tried for
greater values of $n$\ and get
\[%
\begin{array}
[c]{lll}%
M^{\ast}(20)=0.89975, & M^{\ast}(21)=0.89002, & M^{\ast}(22)=0.88043\\
M^{\ast}(23)=0.87101, & M^{\ast}(24)=0.86177,\;\ldots,& M^{\ast}(30)=0.81064.
\end{array}
\]
The last value was obtained in a very long time, so for greater values of
$n$\ we generated some random words $w_{1}$, $w_{2}$,..., $w_{\ell}$ of length
$100$, respectively $200$, $300$, $400$ and $500$ over $A=\{0,1\}$ and get
some roughly approximate values $M^{\ast}(n)\simeq\left(  \mathrm{pal}_{w_{1}%
}(n)+...+\mathrm{pal}_{w_{\ell}}(n)\right)  /\ell$. For $\ell=200$\ we obtained
\[%
\begin{array}
[c]{lll}%
M^{\ast}(100)\simeq0.53, & M^{\ast}(200)\simeq0.39, & M^{\ast}(300)\simeq
0.32,\\
M^{\ast}(400)\simeq0.29, & M^{\ast}(500)\simeq0.26. &
\end{array}
\]
This method allows us to obtain the previous exactly computed values $M^{\ast
}(20),$ ..., $M^{\ast}(30)$\ with two exact digits. These numerical results
allow us to formulate the following

\textbf{Conjecture} The sequence $M_{q}(n)/n$ is strictly decreasing for
$n\geq7$.

\end{document}